%% file: lqptoric.tex
\documentclass[a4paper]{article}
\usepackage{authblk}
\usepackage{latexsym,amsmath,amssymb,enumerate,mathrsfs,amsthm}
\usepackage{graphicx}
\usepackage[bottom]{footmisc}% places footnotes at page bottom

\newcommand{\alg}[1]{\mathfrak{#1}}

\newcommand{\mc}[1]{\mathcal{#1}}

\newcommand{\Hom}{\operatorname{Hom}}
\newcommand{\id}{\operatorname{id}}

\newcommand{\ket}[1]{\ensuremath{\left|#1\right\rangle}}

\newcommand{\supp}{\operatorname{supp}}
\newcommand{\lspan}{\operatorname{span}}
\newcommand{\inv}[1]{\overline{#1}}
\newcommand{\diag}{\operatorname{diag}}

\numberwithin{equation}{section}

\theoremstyle{plain}
\newtheorem{theorem}{Theorem}
\numberwithin{theorem}{section}
\newtheorem{lemma}[theorem]{Lemma}
\newtheorem{proposition}[theorem]{Proposition}

\newtheorem{conjecture}[theorem]{Conjecture}

\newtheorem{definition}[theorem]{Definition}

\title{Kitaev's quantum double model from a local quantum physics point of view}
\date{}

\author{Pieter Naaijkens\footnote{Contribution to R. Brunetti, C. Dappiaggi, K. Fredenhagen, J. Yngvason (eds), \emph{Advances in Algebraic Quantum Field Theory}, Springer (2015).}}

\affil{Institut f{\"u}r Theoretische Physik, Leibniz Universit{\"a}t Hannover, Appelstra{\ss}e 2, 30167~Hannover, Germany, E-mail: {\rm\texttt{pieter.naaijkens@itp.uni-hannover.de}}}

\renewcommand\footnotemark{}

\begin{document}
\maketitle       
\begin{abstract}
A prominent example of a topologically ordered system is Kitaev's quantum double model $\mc{D}(G)$ for finite groups $G$ (which in particular includes $G = \mathbb{Z}_2$, the toric code). We will look at these models from the point of view of local quantum physics. In particular, we will review how in the abelian case, one can do a Doplicher-Haag-Roberts analysis to study the different superselection sectors of the model. In this way one finds that the charges are in one-to-one correspondence with the representations of $\mc{D}(G)$, and that they are in fact \emph{anyons}. Interchanging two of such anyons gives a non-trivial phase, not just a possible sign change. The case of non-abelian groups $G$ is more complicated. We outline how one could use amplimorphisms, that is, morphisms $\alg{A} \to M_n(\alg{A})$ to study the superselection structure in that case. Finally, we give a brief overview of applications of topologically ordered systems to the field of quantum computation.
\end{abstract}

\section{Introduction}
A fundamental result in local quantum physics is that in high enough dimensional space-times (or, for well enough localised particles), charged particles are either \mbox{(para-)}bosons or (para-)fermions~\cite{MR0297259}. This is no longer true in lower dimensional space-times~\cite{FRS1}. Instead of a representation of the symmetric group, representations of the braid group may be obtained from interchanging identical particles. Even though our world appears to have three spatial dimensions, many systems effectively behave like two or one dimensional systems, opening the possibility that they may be described by a low dimensional effective theory, with excitations with braided statistics. 

Interest in such systems has sparked in recent years, getting attention from the theoretical and mathematical physics communities, condensed matter physicists, quantum information theorists, and mathematicians. Part of the reason for this is the discovery of topological quantum computation (first proposed independently by Kitaev~\cite{MR1951039} and Freedman~\cite{MR1612425}), a field at the intersection of quantum theory, quantum computation and quantum topology~\cite{Wang}. One of the attracting features is that the use of topological properties of the system can lead to much better fault-tolerance with respect to (local) perturbations. An overview and candidates for systems that can be used for topological quantum computation can be found in the review~\cite{MR2443722}.

On the mathematical level the behaviour of the anyons, the quasi-particle excitations with braided statistics, can be captured by the concept of a \emph{braided fusion category}. This category encodes all information on how two anyons can combine (``fusion''), what happens if we interchange them (``braiding''), which types can exist in the system, and so on. This structure is in fact very well known in local quantum physics. One of the highlights of algebraic quantum field theory is the study of superselection sectors initiated by Doplicher, Haag and Roberts~\cite{MR0297259,MR0334742}, see also this volume~\cite{FredenhagenLQP,RehrenCFT}. This leads in a natural way to a fusion category as above. Also in the local quantum physics approach to conformal field theories, such categories appear. One can show that for \emph{rational} conformal field theories on the circle, this category is in fact \emph{modular}~\cite{MR1838752}. A modular tensor category is a special type of fusion categories that is highly non-degenerate. It is precisely such anyon models that are of interest to topological quantum computing.

Many of the models that have such anyonic excitations are what is called \emph{topologically ordered}. Topological order is a new type of order that does not fall into the Landau theory of spontaneous symmetry breaking, and there is no local order parameter distinguishing different phases. The examples known are not relativistic theories, and most of them are not described by a conformal field theory (at least, not directly). For example, the model that we will study here, Kitaev's quantum double model, is a quantum spin model, defined on a lattice. Nevertheless, one could try to apply the ideas of Doplicher, Haag and Roberts to these systems, to study the superselection sectors and the properties of the anyons. This is indeed possible~\cite{haagdouble,toricendo}.

This approach gives a strong mathematical founding to the study of anyons in topologically ordered systems. Not only does it allows us to borrow techniques from local quantum physics, it also opens up the way to the use of, for example, operator algebraic methods (see~\cite{jklindex} for such an application). A rigorous mathematical framework may also be of use in this fast-moving field, where no consensus on the right definitions for fundamental concepts is reached. A precise setting where operator algebraic methods can be used may be the right setting to study important questions such as about the stability of such systems: suppose that the Hamiltonian of the system is perturbed by a suitable perturbation, how much of the structure remains? One expects that the topological nature of the system will preserve the interesting properties (as long as, for example, the Hamiltonian remains gapped), but this is something that one wants to prove rigorously (however, see e.g.~\cite{BravyiHastings,BravyiHastingsMichalakis} for results in this direction).

The remainder part of this contribution is outlined as follows. First, we discuss the basic idea behind topological order, and introduce the quantum double model in the setting of local quantum physics. Section~\ref{sec:haagduality} discusses the technical property of Haag duality in this context. We then come to the main part: an overview of the DHR-type analysis of the superselection sectors in the quantum double model. There we also relate it to the theory of (modular) tensor categories. This construction only works for the \emph{abelian} quantum double models. Therefore, in the following section (which contains new results) we outline how one could use so-called amplimorphisms to study the non-abelian case. Finally, in the last section we briefly comment on applications to quantum computing.

\section{Topological order}
Around the late '80s, it was found that there are states of matter that do not fall into the Landau theory of spontaneous symmetry breaking. One of the first papers discussing this is~\cite{PhysRevLett.59.2095}. Such states were called \emph{topologically ordered}, because some of their properties depend on the topology of the manifold on which the system is defined~\cite{PhysRevB.40.7387}. As a particular example one can think of two dimensional systems defined on a surface. For topologically ordered systems, the ground state degeneracy typically depends on the genus of the surface.

Around the change of the millennium, quantum information theorists started to take an increased interest in topologically ordered systems. One of the main reasons is that it was realised that the topological properties of such systems might be useful for quantum computation. This was first suggested by Kitaev~\cite{MR1951039}, who introduced a simple quantum system, the \emph{toric code}, that could act as a quantum memory, storing a pair of qubits. A qubit is the quantum analogue of a bit in a classical computer. Mathematically it is nothing but a copy of the Hilbert space $\mathbb{C}^2$, whose basis is usually denoted by $\ket{0}$ and $\ket{1}$. Although the basics of quantum computing are not difficult to explain, we refer to the standard textbook~\cite{MR1796805} for more information. In any case, the idea is to embed the two qubit Hilbert space $\mathbb{C}^2 \otimes \mathbb{C}^2$ into the physical Hilbert space $\mathcal{H}$. We can then initialize our system in one of the states. The point of a good quantum memory is then that after some time, we should be able to recover the state (at least through the gathering of measurement statistics). In practice, however, quantum systems are not completely isolated, and influences from the environment might drive the system into a different state.

This is where topologically ordered systems come into play. A nice feature of topologically ordered systems is that there is no local\footnote{For finite systems, ``local'' means small compared to the system size.} order parameter. Hence, local observables cannot distinguish between different ground states, and on the other hand local operations cannot put the system from one ground state into another. Hence the idea is to encode information in the ground space, since this will be robust, at least against local perturbations. Of course, nature does not make it easy for us, and there are some more subtle points to being a good quantum memory. For example, because the information is stored non-locally, accessing it or acting on it also necessarily requires non-local operations, which may be difficult to implement. In addition, although local perturbations of the system do not destroy the information, this requires the use of an error correcting protocol. One should be able to do this fast enough, to prevent the errors from spreading out over the system to a \emph{non-local} error, which \emph{does} corrupt the information in the memory. At least for a wide class of 2D systems, it turns out that these systems by themselves are not good memories, so one indeed has to do some error correction~\cite{AlickiThermal,BravyiTerhal,KayColbeck,LandonCarinal-Poulin}.

Another interesting aspect of topologically ordered systems is that they generally have anyonic excitations. That is, excitations of such systems behave like quasi-particles, with anyonic or braided statistics: the global state of the system changes non-trivially (that is, differently from just a sign change) under the interchange of two such quasi-particles. It is this property that we will focus on here. In particular, we will outline how the Doplicher-Haag-Roberts analysis of superselection sectors in algebraic quantum field theory can be translated to the setting of topologically ordered systems, and how this can be used to recover the statistics of the anyons. To illustrate how this works we consider Kitaev's \emph{toric code}, and more generally his quantum double model~\cite{MR1951039}, which is the prototypical example of a topologically ordered system. 

Since the toric code is relatively simple and nevertheless has many interesting features, one can use it as a testbed to see if ideas from algebraic quantum field theory can be applied to this class of models. This indeed turns out to work very well, and allows one to explicitly realize many of the fundamental concepts in the theory of superselection sectors. This shows that the concept can be transferred from relativistic theories to non-relativistic condensed matter systems. Although the explicit constructions we present here depend on the specific knowledge of the toric code, these features are common to a range of topologically ordered systems, and the theory can in principle be applied to them as well.

To make the connection between the toric code and the theory of superselection sectors, it is convenient to depart from the usual setting of Kitaev's model and define it on an \emph{infinite} 2D lattice, instead of on some compact surface of genus $g$. This is most conveniently done in the operator algebraic framework, were one assigns algebras of local observables to finite regions of space~\cite{MR887100,MR1441540}. In particular, there is a clear distinction between local and global observables, just like in local quantum physics. That is, one does not have to keep track of the system size. Concretely, we consider the lattice $\mathbb{Z}^2$, and write $\Gamma$ for the set of \emph{edges} or \emph{bonds} between them. We give each edge an orientation. For simplicity, all vertical edges are assumed to point upwards, the horizontal edges will point to the right (see figure~\ref{fig:ribbon}). Now let $G$ be a finite dimensional group. Then to each edge we assign a ``$G$-spin''. This means that there is a quantum system at each edge, with corresponding Hilbert space $\mathcal{H}_e := \mathbb{C}[G]$, where the right hand side is the group algebra of $G$, seen as a Hilbert space in the natural way.

The local operators that act on an edge $e$ are $\alg{B}(\mathcal{H}_e) \cong M_{|G|}(\mathbb{C})$. If $\Lambda \in \mathcal{P}_f(\Gamma)$, the set of finite subsets of $\Gamma$, then the local observables associated to $\Lambda$ are defined as $\alg{A}(\Lambda) := \bigotimes_{e \in \Lambda}$. If $\Lambda_1 \subset \Lambda_2$ there is a natural inclusion of the corresponding algebras, by tensoring with the identity operator at the sites of $\Lambda_2$ that are not in $\Lambda_1$. Hence we are in the familiar setting of local quantum physics, where we have a net $\Lambda \mapsto \alg{A}(\Lambda)$ of observables associated to bounded regions (cf.~\cite{FredenhagenLQP}). It is not difficult to see that this net is local, that is, $[\alg{A}(\Lambda_1), \alg{A}(\Lambda_2)] = \{0\}$ if $\Lambda_1 \cap \Lambda_2 = \emptyset$.

We then proceed in the same way as for relativistic systems: the local operators are defined as $\alg{A}_{loc} = \bigcup_{\Lambda \in \mc{P}_f(\Gamma)} \alg{A}(\Lambda)$. Note that $\alg{A}_{loc}$ is a $*$-algebra in a natural way, and that the standard operator norm on matrix algebras induces a $C^*$-norm on $\alg{A}_{loc}$. The quasilocal algebra $\alg{A}$ is the completion of $\alg{A}_{loc}$ with respect to this norm. If $\Lambda \subset \Gamma$ is an arbitrary subset (not necessarily finite), we then define
\[
\alg{A}(\Lambda) = \overline{ \bigcup_{\Lambda_f \in \mc{P}_f(\Gamma \cap \Lambda)} \alg{A}(\Lambda_f)}^{\| \cdot \|}
\]
as the algebra of all observables that can be measured inside the region $\Lambda$.

Finally, note that there is a natural translation symmetry on the system, although we do not need this for our purposes, except for picking out a translationally invariant ground state. This touches upon one of the fundamental differences between relativistic quantum field theory and discrete systems. For the discrete systems there is no Lorentz group or Poincar{\'e} group, which play an important role in relativistic theories. Nevertheless, one can mimic some of these concepts using the translation symmetries mentioned, together with locality estimates for local observables evolved under the time evolution of the system~\cite{MR2217299}. For our purposes these aspects play no role.

\subsection{The quantum double model}\label{sec:qdouble}
So far the description has been completely general. To consider a specific system one has to specify the dynamics of the model. Note that because the model is infinite, the corresponding Hamiltonian generally is unbounded, and hence not an element of $\alg{A}$. Fortunately, in practice there is a lot more structure, and the dynamics are local in a suitable sense. More concretely, for each $\Lambda \in \mc{P}_f(\Gamma)$ one can define a self-adjoint $H_\Lambda$, describing the interactions within that region. Heuristically, the dynamic evolution of an observable $A$ is then obtained as $\alpha_t(A) := \lim_{\Lambda \to \infty} e^{i t H_\Lambda} A e^{-i t H_\Lambda}$, where $\Lambda \to \infty$ means that we take an increasing sequence of finite sets $\Lambda$ that exhaust $\Gamma$. If the strength of the interaction decays fast enough, this expression converges and one obtains a strongly continuous one-parameter group $t \mapsto \alpha_t$ of automorphisms~\cite{MR1441540}.

Once the dynamics are defined one can talk about ground state. The most convenient way to do this is in terms of (generally unbounded) $*$-derivations $\delta$, which are obtained as the generator of an automorphism group $t \mapsto \alpha_t$. In the cases of interest to us these are simply obtained as (the closure of) $\delta: \alg{A}_{loc} \to \alg{A}_{loc}$, with
\[
	\delta(A) := i \lim_{\Lambda \to \infty} [ H_\Lambda, A].
\]
Because the interactions are local and of finite range in the cases of interest, this converges, and $\alpha_t(A) = e^{t \delta}(A)$. A state $\omega_0$ on $\alg{A}$ is then called a \emph{ground state} if $-i \omega_0(A^* \delta(A)) \geq 0$ for all $A \in D(\delta)$. This (at first sight) perhaps strange looking condition can be interpreted as a positivity of the energy condition. Such a state is automatically invariant, $\omega_0 \circ \alpha_t = \omega_0$, so that in the GNS representation the dynamics are implemented by a strongly continuous group of unitaries $t \mapsto U(t)$. By Stone's theorem one obtains a Hamiltonian $H$, which can be shown to be positive using the condition mentioned above~\cite{MR1441540}.

There is some extra notation that has to be introduced to define the dynamics in Kitaev's quantum double model. A combination $s = (v,f)$ of a vertex $v$ and a choice of an adjacent face $f$ is called a \emph{site}. To each site we associate the following operators. Let $g \in G$. Note that for a vertex $v$, there are four edges that start or end in $v$. A basis for the Hilbert space of this four sites is given by $\ket{g_1} \otimes \cdots \otimes \ket{g_4}$, with $g_i \in G$. The operator $A^g_s$ acts on this basis vector by multiplying $g_i$ from the \emph{left} with $g$ if the corresponding edge points away from the vertex $v$, and $g_i$ gets sent to $g_i g^{-1}$ if the edge points inwards. If $h \in G$, a projection $B^h_s$ is defined as follows. First, list the edges around the face $f$ in counter-clockwise order, starting in $v$. On the basis labelled by group elements of the Hilbert space corresponding to these four edges, $B_s^h$ acts as the identity if $\sigma(g_1) \cdots \sigma(g_4) = h$, and zero otherwise. Here $\sigma(g) = g$ if the corresponding edge is in the same direction as the counter-clockwise labelling, and $\sigma(g) = g^{-1}$ otherwise. Pictorially the operators can be depicted as follows:\\
\begin{center}
	\includegraphics{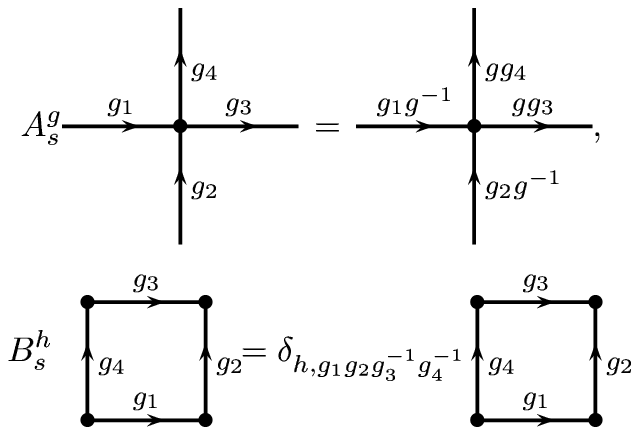}
\end{center}
Sometimes one says that $B^h_s$ projects on the states with \emph{flux} $h$ through the face $f$. Indeed, the structure is very similar to that of lattice gauge theory~\cite{Oeckl}, and the quantum double model can in fact also be interpreted in such terms~\cite{MR1951039}.

The operators $A_s$ and $B_s$ are also called \emph{star} and \emph{plaquette} operators respectively, and satisfy the following algebraic relations, which can easily be verified:
\begin{equation}
	A_s^g A_s^h = A_s^{gh}, \quad B^g_s B^h_s = \delta_{g,h} B_s^h, \quad A^g_s B^h_s = B^{g h g^{-1}}_s A^g_s, \quad (A^g_s)^* = A_s^{g^{-1}}.
	\label{eq:starplaq}
\end{equation}
Star and plaquette operators acting on different sites $s$ and $s'$ always commute. At this point the name ``quantum double'' can be explained: the algebraic relations above are exactly those of the quantum double $\mathcal{D}(G)$ of the Hopf algebra $\mathbb{C}[G]$ (see e.g.~\cite{MR1321145} for an introduction). That is, at each site there is an action of the Hopf algebra which acts via the star and plaquette operators. The representation theory of $\mathcal{D}(G)$ plays an important role in the analysis of the superselection sectors, as we will see below.

We are now in a position to define the dynamics of the system. First, write $A_s = \frac{1}{|G|} \sum_{g \in G} A^g_s$ for the average of the star operators, and $B_s := B^e_s$. Then $[A_s, B_s] = 0$. The local Hamiltonians are then defined as
\begin{equation}
	\label{eq:hamiltonian}
	H_\Lambda = - \sum_{s \in \Lambda} A_s - \sum_{s \in \Lambda} B_s,
\end{equation}
where in the first sum the summation is over all sites $s = (v,f)$, such that the edges starting or ending at $v$ are contained in $\Lambda$. Similarly, the second summation is over all sites such that the edges of the face $f$ are all contained in $\Lambda$. The case $G=\mathbb{Z}_2$ is called the \emph{toric code}. These local Hamiltonians generate a one-parameter group of automorphisms $\alpha_t$ as described above. It turns out that there is a unique \emph{translationally invariant} ground state~\cite{MR2345476,phdthesis}.\footnote{Bruno Nachtergaele pointed out that the remark in~\cite{toricendo} is in fact false, and there are additional (non-translationally invariant) ground states. Indeed, the charged states constructed in~\cite{toricendo} have dynamics implemented by a positive Hamiltonian.}
\begin{proposition}
There is a unique translationally invariant ground state $\omega_0$. This state is pure and is the unique state on $\alg{A}$ which satisfies $\omega_0(A_s) = \omega_0(B_s) = 1$ for all star and plaquette operators $A_s$ and $B_s$.
\end{proposition}
This ground state will be the starting point from the analysis: the different superselection sectors will be realized, roughly speaking, by creating single excitations of this ground state. The GNS representation corresponding to the ground state representation from the proposition will be denoted by $(\pi_0, \Omega, \mc{H}_0)$. Since $\alg{A}$ is a UHF (and hence simple) algebra, $\pi_0$ is automatically injective, and we will often identify $\pi_0(A)$ with $A$.

To understand how this works it is necessary to understand what we mean by an excitation, and how these can be obtained. From the proposition above it follows that in the ground state representation, $A_s \Omega = B_s \Omega = \Omega$. We can interpret this as constraints, and violating some of these constraints carry an energy penalty, according to the local Hamiltonians~\eqref{eq:hamiltonian}. We will interpret the violation of such a constraint as a (quasi-)particle sitting at the site $s$ (if the state is an eigenstate of $H_\Lambda$).\footnote{Note that we disregard any momentum variables. The pairs of excitations form bound states, so the excitations mentioned here are not quite the single-particle excitations one encounters in scattering theory.} So the excitations live at sites of the lattice. 

\begin{figure}
	\begin{center}
	\includegraphics{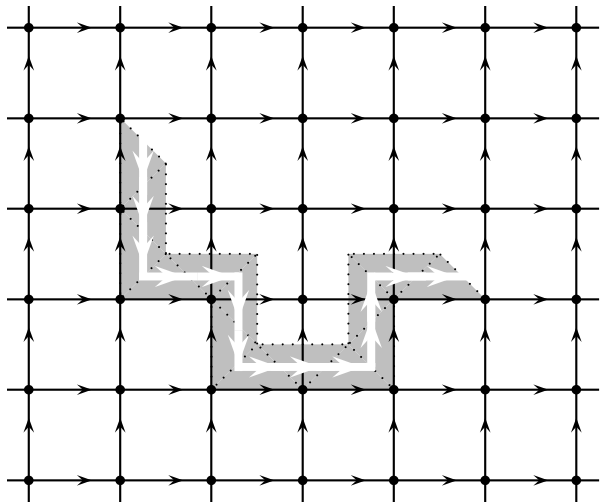}
	\end{center}
	\caption{The lattice describing the system, together with a ribbon between two sites. Note that the ribbon has an orientation, pointing from start to end.}
	\label{fig:ribbon}
\end{figure}

The excitations can be obtained by acting with so-called \emph{ribbon operators} on the ground state. We only list their main properties here. A more complete treatment and proofs can be found in~\cite{PhysRevB.78.115421}. A \emph{ribbon} is, roughly speaking, a continuous path of triangles (see figure~\ref{fig:ribbon} for an example). To each such ribbon $\xi$ and pair $g,h$ of group elements one can assign a ribbon operator $F^{h,g}_\xi$. These operators act only on the edges that are part of the ribbon (or cross any of the triangles). They satisfy the following algebraic relations:
\[
F_{\xi}^{h_1, g_1} F_{\xi}^{h_2, g_2} = \delta_{g_1, g_2} F^{h_1 h_2, g_1}_\xi, \quad (F^{h,g}_{\xi})^* = F^{h^{-1}, g}_{\xi}, \quad \sum_{g \in G} F^{e,g}_{\xi} = I.
\]
Note that the ribbons carry an orientation. This allows us to talk about the \emph{starting} and \emph{ending} sites of the ribbon. Let us denote them by $s_1$ and $s_2$ for the moment. An important property of the ribbon operators is that they commute with any star or plaquette operator, except for those at the ending sites of the ribbon. There we have the following commutation relations:
\begin{equation}
	\label{eq:comms}
	A^k_{s_1} F^{h,g}_\xi = F^{k h k^{-1}, kg}_\xi A^k_{s_1}, \quad B^k_{s_1} F^{h,g}_\xi = F^{h,g}_\xi B^{kh}_{s_1}.
\end{equation}
Similarly, for the ending site 
\begin{equation}
	\label{eq:comme}
	A^k_{s_2} F^{h,g}_\xi = F^{h , g k^{-1}}_\xi A^k_{s_2}, \quad B^k_{s_1} F^{h,g}_\xi = F^{h,g}_\xi B^{g^{-1} h^{-1} g k}_{s_2}.
\end{equation}
From these commutation relations it follows that the vector state $F^{h,g}_\xi \Omega$ violates some of the constraints the Hamiltonian gives for the ground state, namely precisely those at the start- and ending sites. Hence this vector state can be thought of as an excited state of the system. These excited states satisfy the following important property, whose proof we will omit (see e.g.~\cite{PhysRevB.78.115421,haagdouble} for the proof).
\begin{lemma}
	Let $\xi_1$ and $\xi_2$ be two ribbons with the same starting and ending sites. Then $F_{\xi_1}^{h,g} \Omega = F_{\xi_2}^{h,g} \Omega$.
\end{lemma}
In other words, the state does not depend on the ribbon itself, only on the endpoints. This already shows the topological nature of the system.

Now suppose that $\xi$ is a non-trivial ribbon. It can then be written as the concatenation of ribbons $\xi_1$ and $\xi_2$. The ribbon operator $\xi$ is related to the smaller ribbons by means of the following recursion relation (which indeed can also be used to recursively define the ribbon operators):
\begin{equation}
	F_{\xi}^{h,g} = \sum_{k \in G} F_{\xi_1}^{h,k} F^{\inv{k} h k, \inv{k} g}_{\xi_2}.
	\label{eq:ribbonop}
\end{equation}
This feature will be important later.

The basis $F^{h,g}_\xi$ of ribbon operators is not always the most convenient one. Recall that for each site $s$ there is an action of $\mc{D}(G)$ on the Hilbert space. One can now consider the vector space $V = \lspan \{ F_\xi^{h,g}\Omega : h,g \in G \}$ for a fixed ribbon $\xi$. By equations~\eqref{eq:comms} this space has a natural interpretation as a left $\mc{D}(G)$-module, because $A^g_s \Omega = \Omega$ and $B^h_s \Omega = \delta_{h,e} \Omega$. Similarly, equations~\eqref{eq:comme} also give it the structure of a $\mc{D}(G)$-module, induced by the action of the star and plaquette operators at the ending site.\footnote{This is in fact the contragradient module of the representation at the starting site.} This suggests that it may be useful to decompose this space into irreducible subspaces, and find a new basis of the ribbon operators accordingly. This indeed turns out to be a good idea. 

The representation theory of $\mc{D}(G)$ is well understood: constructions of all irreducible representations can be found in~\cite{MR1797619,DijkgraafOrbifold}. They are in one-to-one correspondence with pairs $(C, \rho)$, where $C$ is a conjugacy class of $G$ and $\rho$ an irreducible representation of $Z_G(c)$, where $Z_G(c)$ is the centralizer in $G$ of a representative $c$ of the conjugacy class. It turns out to be convenient to label the elements of the conjugacy class in a particular way (compare \cite{PhysRevB.78.115421}). First of all, let $C$ be a conjugacy class of $G$. Choose a representative $r \in C$, and let $Z_G(r)$ be the centraliser of $r$ in $G$. We label the elements of $C$ by $c_1, \dots, c_n$, where $n = |C|$. Then there are $q_i$ such that $c_i = q_i r \inv{q}_i$, where we used the notation $\inv{g}$ for the inverse of $g$, to improve readability. The set $\{q_i\}$ is denoted by $Q_C$. Note that each $g \in G$ can be uniquely written as $g = q_i n$ for some $q_i \in Q_C$ and $n \in Z_G(r)$.

With this notation we can choose a new basis of the ribbon operators associated to a ribbon $\xi$. Suppose that $\rho$ is a unitary representation of $Z_G(r)$. We regard each $\rho(g)$ as a unitary matrix, and write $\rho(g)_{jj'}$ for the corresponding matrix elements in the standard basis. Let $i,i' = 1, \dots n$ and $j, j' = 1, \dots \dim(\rho)$. We then define  
\[
F^{C\rho;i,i',j,j'}_\xi = \sum_{g \in Z_G(r)} \overline{\rho}_{jj'}(g) F_\xi^{\overline{c}_i, q_i g \overline{q}_{i'}}.
\]
As $C$ runs over all conjugacy classes of $G$, and $\rho$ runs over the corresponding irreducible representations of the centralisers, these operators form a basis of the space spanned by $F^{h,g}_\xi$. We refer to~\cite{PhysRevB.78.115421} for a proof. In essence, the point is that the space of operators is decomposed into subspaces transforming according to some irreducible representation of $\mc{D}(G)$. For convenience we sometimes drop the notation for $C$ and $\rho$ when these are implied by the context, and write $I = (i,j)$, $J = (i', j')$ for the pairs of indices: $F_{\xi}^{IJ}$.

If the group $G$ is abelian the notation can be simplified, since each conjugacy class consists of exactly one element, and clearly the centralizer $Z_G(r)$ is always equal to $G$. Hence we can label the basis by a pair $(\omega,c)$, where $\omega$ is a character of $G$ and $c \in G$. It can be checked that in that case the corresponding ribbon operators $F^{\omega,c}_\xi$ are unitaries and that $F_\xi^{\omega,c} F_\xi^{\chi,d} = F^{\omega \chi, cd}_\xi$. There are many more useful relations, describing for example the commutation relations between two ribbon operators acting on crossing ribbons, these can be found in~\cite{PhysRevB.78.115421}.

\section{Haag duality}\label{sec:haagduality}
Just as in the Doplicher-Haag-Roberts theory of superselection sectors, Haag duality plays a very useful role. Let us first recall what Haag duality actually is: it is an commutation property of von Neumann algebras generated by algebras of observables in the ground state representation. More precisely, let $\Lambda$ be a cone-like region.\footnote{The precise shape is not so important, but see~\cite{haagdouble} for a precise definition. Heuristically, one can take a point in the lattice, draw two semi-infinite lines from this point, and consider all edges that either intersect these lines, or lie in the convex set bounded by the lines.} Then $\pi(\alg{A}(\Lambda))''$ is the von Neumann algebra generated by all (quasi-)local observables localised in $\Lambda$. Note that, by locality, these observables commute with any local observable localised \emph{outside} the cone, that is, in $\Lambda^c$. In other words, $\pi_0(\alg{A}(\Lambda))'' \subset \pi_0(\alg{A}(\Lambda^c)'$, where $\Lambda^c$ is the complement of $\Lambda$ in $\Gamma$. Haag duality says that these sets are actually equal. It means that we cannot add operators to the cone algebra $\pi_0(\alg{A}(\Lambda))''$ without violating locality. This property is fulfilled for Kitaev's model for abelian groups $G$:

\begin{theorem}
Let $G$ be a finite abelian group and let $\Lambda$ be a cone. Write $\pi_0$ for the translational invariant ground state representation of the quantum double model for $G$. Then Haag duality holds:
\[
	\pi_0(\alg{A}(\Lambda))'' = \pi_0(\alg{A}(\Lambda^c))',
\]
where the prime denotes the commutant in the set of all bounded operators.
\end{theorem}
This theorem was first proven in~\cite{haagdtoric} for the case of $G = \mathbb{Z}_2$ (the toric code) and later extended to all finite abelian groups in~\cite{haagdouble}.

The proof of the theorem depends on a good knowledge of the pairs of excitations and the operators that generate them. This makes it possible to find a convenient description of the Hilbert space of the system, and the Hilbert space $\mathcal{H}_\Lambda$ describing pairs of excitations localised in a fixed cone $\Lambda$. With this description it is possible to demonstrate that the self-adjoint parts of certain algebras (the restriction to $\mathcal{H}_\Lambda$ of the algebras of observables localised in $\Lambda$, and the algebra of observables localised in $\Lambda^c$), when acting on the cyclic GNS vector generate a dense subset of $\mathcal{H}_\Lambda$. It is known that this is related to commutation properties of algebras~\cite{MR0383096}, and this allows one to conclude Haag duality.

Much of the proof boils down to a good understanding of the representation theory of $\mathcal{D}(G)$, the quantum double of $G$, since the operators that create pairs of excitations correspond to irreducible representations of the Hopf algebra. Their behaviour under the interchange of two such operators (or on the representation theory side, the braiding of the category of representations) also plays a role. These properties are well-known and studied for a wider class of quantum doubles, particularly those associated to Hopf-$*$ algebras that are quasi-triangular (which means a braiding can be defined). As Kitaev already remarked in his paper~\cite{MR1951039}, the quantum double model can also be defined for such Hopf algebras. The following conjecture therefore seems very natural: 
\begin{conjecture}
Consider a quasi-triangular Hopf-$*$ algebra $H$. Then the corresponding Kitaev model on the plane satisfies Haag duality for cones.
\end{conjecture}
The main difficulty in proving this is that the combinatorics get much more involved. In particular, in the case of abelian algebras all irreducible representations are one dimensional. For non-abelian algebras one has to consider multiplets of operators that create pairs of excitations, transforming according to some irreducible representation of the Hopf algebra $\mathcal{D}(H)$. We will see some of the consequences later on, when we discuss non-abelian theories.

It should be noted that many of the constructions we present below do not depend on Haag duality. The only place where it is used is in going from a representation that satisfies the selection criterion to a localised endomorphism. Representatives of these endomorphisms (and intertwiners between them that change the localisation region) can be constructed without an appeal to Haag duality. However, it is then not possible to conclude that each representative of the equivalence class is indeed given by a localised endomorphism.

\section{Superselection theory for abelian models}
It is well-known that the superposition principle of quantum mechanics does not hold unrestrictedly. A familiar example is that of bosons and fermions~\cite{WWW}. Consider the vector $\psi = \frac{1}{\sqrt{2}}( \psi_f + \psi_b)$, where $\psi_f$ is a single-particle fermionic state, while $\psi_f$ is bosonic. Under a rotation of the system by 360 degrees, the fermionic state will acquire a minus sign, while the bosonic part is unchanged. Physically, however, the states are indistinguishable, and will lead to the same expectation values. In general, one can argue that that the physical Hilbert space can be decomposed into different sectors corresponding to the different types of charges in the system. This is called a \emph{superselection rule}.

On a mathematical level, superselection sectors arise because there are inequivalent representations of the observable algebra $\alg{A}$. A $C^*$-algebra has many inequivalent representations in general, hence one should somehow select the physically relevant representations. There are several conditions one might impose. For example, in relativistic theories a natural condition is to look at representations that are covariant with respect translations, such that the spectrum of the generators (that is, the momentum) is contained in the forward light cone. Doplicher, Haag and Roberts proposed to look at those representations that, roughly speaking, look like the vacuum representation in the spacelike complement of a double cone, the intersection of a forward and backward light-cone~\cite{MR0297259}. Although this certainly does not cover all physical systems of interests (it excludes, for example, electromagnetic charges~\cite{MR660538}), it is nevertheless a useful criterion, and can for example also be applied to conformal theories on the circle~\cite{GFCFT}.

What Doplicher, Haag and Roberts (DHR) showed~\cite{MR0297259,MR0334742} (see also~\cite{FredenhagenLQP,RehrenCFT}) is that these representations can be studied in a systematic way. In this way, one can learn something about the statistics of the different charges (i.e., if they are bosons or fermions in space-times of dimensions $\geq 2+1$), or how they behave under composition of two charges (so-called \emph{fusion rules}). Mathematically, it amounts to studying the structure of the equivalence classes satisfying the selection criterion as a \emph{braided fusion category}, which is a type of \emph{tensor category}. Some of the essential steps will be outlined below, but a more thorough introduction can be found in, e.g.,~\cite{MR1405610,halvapp}. The study of tensor categories is a whole field on its own (an introduction can be found in~\cite{muegerfusion}), but familiarity with the main terminology is not strictly necessary for our purposes here.
 
In the remainder of this section $G$ will be a finite abelian group. We will see how one can construct different superselection sectors for Kitaev's quantum double model, and do a Doplicher-Haag-Roberts type of analysis to recover the properties of the different charges. The results outlined here have been obtained in~\cite{haagdouble,toricendo}, to which the reader is referred for details. It is perhaps somewhat surprising to see that this theory, which was originally developed for relativistic quantum systems, can be applied so successfully to discrete lattice quantum spin systems, even though there are fundamental technical differences.

\subsection{Localized representations}
As mentioned, the different superselection sectors are identified as equivalence classes of representations of the observable algebra, satisfying additional selection criteria. Here we take the ground state representation $\pi_0$, corresponding to the translationally invariant ground state, as a reference representation, and look at all representations that look like $\pi_0$ when considering observables outside a \emph{cone-like region} $\Lambda$. The reason for this will become clear later, but at this point we mention that this is similar to the work of Buchholz and Fredenhagen, who show that in relativistic theories, massive particles can be localised in \emph{spacelike} cones~\cite{MR660538,BF-ICMP1981}, leading to a similar criterion. We will write $\mc{L}$ for the set of cones.

\begin{definition}\label{def:select}
A representation $\pi$ of $\alg{A}$ is called \emph{localizable} if for each $\Lambda \in \mc{L}$ we have
\begin{equation}
	\pi_0 \upharpoonright \alg{A}(\Lambda^c) \cong \pi \upharpoonright \alg{A}(\Lambda^c),
	\label{eq:select}
\end{equation}
where $\cong$ denotes unitary equivalence of representations and $\upharpoonright$ means that we restrict the representation to the subalgebra $\alg{A}(\Lambda^c)$.
\end{definition}
An equivalence class of representations satisfying the criterion is called a \emph{(superselection) sector} or simply a \emph{charge}. 

To proceed in the analysis of these representations, we have to pass from representations to endomorphisms of the observable algebra. This can be done with the help of Haag duality. Indeed, let $\Lambda \in \mc{L}$. Then from equation~\eqref{eq:select} there is a unitary $V_\Lambda$ setting up the equivalence with $\pi_0$ for the subalgebra $\alg{A}(\Lambda^c)$. Define $\rho(A) = V_\Lambda \pi(A) V_\Lambda^*$ and let $\Lambda_2 \in \mc{L}$ contain $\Lambda$. Let $A \in \alg{A}(\Lambda_2)$ and $B \in \alg{A}(\Lambda_2^c)$. Note that $\rho(B) = \pi_0(B)$. Moreover, by locality,
\[
\rho(AB) = \rho(A) \pi_0(B) = \rho(BA) = \pi_0(B)\rho(A),
\]
hence by Haag duality $\rho(A) \in \pi_0(\alg{A}(\Lambda_2^c))' = \pi_0(\alg{A}(\Lambda_2))''$. Since $\pi_0$ is faithful, we can identify $\alg{A}$ with $\pi_0(\alg{A})$ and regard $\rho$ as an endomorphism of $\alg{A}$.\footnote{Strictly speaking, this is only true if $\alg{A}(\Lambda_2)'' \subset \alg{A}$, which in general is not the case for unbounded regions $\Lambda_2$. One can however solve this by passing to a larger algebra $\alg{A}^{\Lambda_a} \supset \alg{A}$ and extend $\rho$ to a proper endomorphism of $\alg{A}^{\Lambda_a}$~\cite[Sect. 4]{MR660538}.} The endomorphism is called \emph{localised} in $\Lambda$. It is also \emph{transportable}: if $\Lambda' \in \mc{L}$, there is a unitary $V$ and endomorphism $\rho'$ (localised in $\Lambda'$) such that $V \rho(A) = \rho'(A) V$ for all $A \in \alg{A}$. Such an operator $V$ (not necessarily unitary) is called an \emph{intertwiner} from $\rho$ to $\rho'$. If $V$ is unitary we will also call it a \emph{charge transporter}, since it moves a charge from one cone to another. Using Haag duality one can show that in fact $V \in \alg{A}(\widehat{\Lambda})$ for suitable $\widehat{\Lambda} \in \mc{L}$ (where $\widehat{\Lambda}$ should contain the localization regions of both $\rho$ and $\rho'$). There is in fact a 1-1 correspondence between superselection sectors and equivalence classes of these localised and transportable endomorphisms. From now on we will work with the latter.

\subsection{Localized sectors in the quantum double model}
Recall that the ribbon operators $F_\xi^{\omega,c}$ create a \emph{pair} of excitations (or charges). We are however interested in the properties of a \emph{single} charge. Since we are in an infinite system, it is possible to create a pair of excitations, and move one of them to infinity. As may already be anticipated from the discussion in Section~\ref{sec:qdouble}, the different charges are in 1-1 correspondence with pairs $(\omega,c)$ of a character $\omega$ and group element $c$. The ribbon operator $F_\xi^{\omega,c}$ then creates the corresponding charge at the beginning of $\xi$, and a conjugate charge at the other end. On the level of the observables, this means that we map $A \mapsto F_\xi^{\omega,c} A (F_{\xi}^{\omega,c})^*$.\footnote{This actually corresponds to creating charges with $(F_\xi^{\omega,c})^*$, which turns out to be slightly more convenient.} Hence to implement the idea of moving one charge to infinity, one can choose a semi-infinite ribbon $\xi$, which for simplicity we assume to be completely inside some $\Lambda \in \mc{L}$, and write $\xi_n$ for the (finite!) ribbon consisting of the first $n$ parts. Then we define
\begin{equation}
	\alpha(A) = \lim_{n \to \infty} F_{\xi_n}^{\omega,c} A (F_{\xi_n}^{\omega,c})^*.
	\label{eq:auto}
\end{equation}
Considering the dense set of local observables, and using the decomposition rule~\eqref{eq:ribbonop}, it is not so difficult to show that this expression converges, and defines an automorphism of $\alg{A}$. It describes how the observables change in the presence of a \emph{single} charge $(\omega,c)$ in the background.

The map $\alpha$ defined above leads us to an example of a representation that satisfies the selection criterion, by defining $\pi = \pi_0 \circ \alpha$. Clearly, if $\Lambda$ is any cone containing the ribbon, $\alpha(A) = A$ for all $A \in \alg{A}(\Lambda^c)$ by locality, so $\alpha$ is localised. It is also transportable. There is a neat way of seeing this: because the vector state $F_{\xi_n}^{\omega,c} \Omega$ depends only on the endpoints of the path, it follows that two for automorphisms $\alpha_1$ and $\alpha_2$ defined in terms of the same charge, but different ribbons (with the same fixed endpoints), the states $\omega_0 \circ \alpha_1$ and $\omega_0 \circ \alpha_2$ coincide. In addition, note that both triples $(\pi_0 \circ \alpha_i, \Omega, \mc{H}_0)$ are GNS triples for this state, so that by the uniqueness of the GNS construction, $\pi_0 \circ \alpha_1$ must be unitarily equivalent to $\pi_0 \circ \alpha_2$. The argument can be extended if the endpoints of the ribbon do not coincide, by conjugating with a suitable ribbon operator. It follows that the automorphisms are transportable. The corresponding intertwiners can be shown to be in $\pi_0(\alg{A}(\Lambda))''$ for a suitable cone $\Lambda$ using Haag duality, but for the quantum double model an explicit construction of a net converging in the weak operator topology to an intertwiner is also possible.

The construction above gives for each pair $(\omega,c)$ an equivalence class of localised and transportable automorphisms, but it is not clear yet that these classes are distinct. This does turn out to be the case. The key idea in showing this is by considering an analogue of Wilson loops, which measure the charge in the region enclosed by the loop. Such operators can be obtained by considering ribbons with the same start and ending sites. In this way one can construct a charge measurement operator in an arbitrarily large region, that have expectation value 1 in the state $\omega_0 \circ \alpha$, where $\alpha$ has charge $(\omega,c)$, and zero expectation value in states obtained from a different pair $(\omega', c')$. It follows that the corresponding representations $\pi_0 \circ \alpha$ must be inequivalent. The discussion can be summarized in the following theorem:
\begin{theorem}
Let $G$ be an abelian group. For each pair $(\omega,c)$ of a character of $G$ and an element $c \in G$, there is an equivalence class of localised and transportable automorphisms. If $\alpha_1$ and $\alpha_2$ are such automorphisms, then $\pi_0 \circ \alpha_1 \cong \pi_0 \circ \alpha_2$ if and only if the belong to the same charge class $(\omega,c)$.
\end{theorem}

This result can also be understood as an instance of \emph{charge conservation}. The total charge of the state can be obtained by finding all excitations $\omega$ (resp. $c$) in the state, and multiply all of them together. Note that this is possible because the dual of an abelian group is a group again. The ribbon operators span a dense subset of the quasilocal algebra $\alg{A}$. Because the two distinct charges at the end of the ribbons transform according to \emph{conjugate} representations, it follows that one cannot change the total charge of the system just by local operations. In particular, one cannot go from a charged sector with a certain total charge to another sector with a different total charge by acting with local operators, so that indeed we have inequivalent representations of $\alg{A}$.

\subsection{Braiding and fusion}
So far we have constructed endomorphisms that describe single charges localised in cones, and charge transporters or intertwiners that can move the charges around. There is however much more structure, and this is where it becomes essential that we have endomorphisms rather than representations: in contrast to representations, endomorphisms can be composed. If $\rho_1$ and $\rho_2$ are two localised endomorphisms, we define $\rho_1 \otimes \rho_2 (A) = \rho_1 \circ \rho_2(A)$. The interpretation is that we first create a charge $\rho_2$, and then add a charge $\rho_1$ to the system. This operation is called \emph{fusion}. Note that $\rho_1 \circ \rho_2$ is localised again, more particularly it is localised in any cone that contains the localization regions of $\rho_1$ and $\rho_2$. It is also transportable again. This can be seen by defining a product operation for intertwiners as well: if $S\rho_1(A) = \rho_2(A)S$ and $T \sigma_1(A) = \sigma_2(A) T$ for all $A \in \alg{A}$, we define $S \otimes T := S \rho_1(T)$. Then an easy calculation shows that $S \otimes T$ intertwines $\rho_1 \otimes \sigma_1$ and $\rho_2 \otimes \sigma_2$. There is however a slight technical issue. The intertwiners $S$ are in general only elements of weak closures of the form $\pi_0(\alg{A}(\Lambda))''$ for some cone $\Lambda$ (this follows from Haag duality), while the endomorphisms are \emph{a priori} only defined on $\alg{A}$. It is however possible to extend the endomorphisms to a slightly larger algebra $\alg{A}^{\Lambda_a}$, where $\Lambda_a$ is some fixed auxiliary cone, that does contain the intertwiners. Since this is a minor technical point, we ignore the issue here, and refer to~\cite{MR660538,toricendo} for technical details.

Even if $\rho_1$ and $\rho_2$ are irreducible, their composition need not be irreducible any more. A natural question therefore is if it is possible to decompose the composition into irreducibles again. The \emph{fusion rules} give this decomposition. That is, if $\rho_i$ and $\rho_j$ are irreducible endomorphisms, there are integers $N_{ij}^k$, where $k$ runs over a set of representatives of all irreducible localised and transportable endomorphisms, such that
\begin{equation}
	\rho_i \circ \rho_j \cong \sum_{k} N^k_{ij} \rho_k.
	\label{eq:fusionrules}
\end{equation}
The sum operation is a direct sum, which can be described in terms of intertwiners. This is very similar to representation theory of finite (or compact) groups: the tensor product of two irreducible representations can be decomposed again as the direct sum of finitely many irreducible representations.\footnote{Indeed, if the theory contains only bosonic sectors, the localised and transportable endomorphisms are in one-to-one correspondence with the representations of a compact global symmetry group $G$~\cite{DR90}.} Note that the identical endomorphism $\iota(A) = A$ acts as a unit for the fusion operation. Finally, a charge $\rho$ has a \emph{dual} or \emph{conjugate} $\overline{\rho}$ if there is such a localised and transportable $\overline{\rho}$ with $N^{\iota}_{\rho \overline{\rho}} = 1$. The existence of conjugates is not automatic, and has to be either proven or taken as an assumption. Physically, it is related to the existence of (charge) anti-particles.

In algebraic quantum field theory, it is usually shown that such a decomposition exists with the help of a technical property, \emph{Property B}~\cite{Borchers67}, that essentially allows any projection in the local algebra to be written in the form $P = W W^*$ for some isometry $W$, localised in a slightly bigger region. In the present situation of spin systems, one could go about by showing that direct sums can be constructed, and by explicitly verifying the fusion rules~\eqref{eq:fusionrules}~\cite{toricendo}. This is particularly easy in abelian models, where there always is a unique fusion outcome. That is, $N^{k}_{ij}$ is equal to one for exactly one value of $k$, and zero otherwise. Symbolically, for the abelian quantum double model the fusion rules are
\[
	(\omega_1, c_1) \otimes (\omega_2, c_2) = (\omega_1 \omega_2, c_1 c_2). 
\]
This can be seen by considering, without loss of generality, a single ribbon $\xi$, and define the automorphisms $\alpha_1$ and $\alpha_2$ corresponding to the choices of charges. Using the multiplication rule for two ribbon operators in abelian models, mentioned at the end of Section~\ref{sec:qdouble}, the result follows.

The final piece of structure that we discuss here is that of \emph{braiding}. This is related to the statistics of identical particles, that is, their behaviour under interchange. This amounts to the study of the relation between $\rho_1 \otimes \rho_2$ and $\rho_2 \otimes \rho_1$ for general localised endomorphisms $\rho_1$ and $\rho_2$. An intertwiner $\varepsilon_{\rho_1,\rho_2}$ relating the two can be constructed explicitly by, quite literally, moving $\rho_2$ around (or, $\rho_1$ of course, in similar way). More precisely, one can choose a localization region $\widehat{\Lambda}$ that is disjoint from the localization regions of both $\rho_1$ and $\rho_2$. By transportability there is a unitary $U$ and a endomorphism $\widehat{\rho}$ localised in $\widehat{\Lambda}$ such that $U \rho_2(A) = \widehat{\rho}(A) U$. Because $\rho_1 \otimes \widehat{\rho} = \widehat{\rho} \otimes \rho_1$ since their localization regions are disjoint, it follows that $\varepsilon_{\rho_1, \rho_2} := (U^* \otimes I)(I \otimes U) = U^* \rho_1(U)$ is an intertwiner from $\rho_1 \circ \rho_2$ to $\rho_2 \circ \rho_1$.

In three or more spatial dimensions, the definition of $\varepsilon_{\rho_1, \rho_2}$ is completely independent of the choices made, and one can show that $\varepsilon_{\rho_1,\rho_2}\varepsilon_{\rho_2,\rho_1} = I$. Note that this is the situation of ordinary bosons and fermions: moving one particle around the other doesn't change the system. One can show that the operators $\varepsilon_{\rho, \rho}$ indeed induce a representation of the permutation group, which interchanges charged excitations~\cite{MR0297259}. In lower dimensional space times, things get more interesting. It is no longer true that there is a unique choice of $\varepsilon_{\rho_1, \rho_2}$~\cite{FRS1,MR1199171,FG90}. The reason is that one cannot continuously move the cone $\widehat{\Lambda}$ around, but instead has to choose between either ``left'' or ``right''. This can be defined unambiguously by choosing an auxiliary cone $\Lambda_a$ as above. To get a consistent definition for the operators $\varepsilon_{\rho_1, \rho_2}$ (for example, to make sure that the different ways to go from ($\rho \otimes \sigma \otimes \tau$ to $\tau \otimes \rho \otimes \sigma$ coincide), one has to choose one of the two alternatives, and stick to that. A consequence of this ambiguity is that it is no longer true that $\varepsilon_{\rho_1,\rho_2} \varepsilon_{\rho_2,\rho_1} = I$. In addition, one does not obtain a representation of the permutation group any more, but rather one of the braid group. Consequently, such charges are said to have \emph{braid statistics}. If this representation is abelian, one calls the charges (abelian) anyons (because they can pick up ``any'' phase under interchange). If the representation is non-abelian, one usually speaks of non-abelian anyons.\footnote{In the local quantum physics literature the name \emph{plektons} is used. This name does not seem to have caught on outside of that community: the name non-abelian anyons prevails for example in the field of topological quantum computing, and this is why we adhere to that name here.}

The charges of the quantum double model are abelian anyons. This can be verified directly, by calculating the operators $\varepsilon_{\rho_1, \rho_2}$ and $\varepsilon_{\rho_2,\rho_1}$, together with the explicit construction of charge intertwiners, mentioned above. In the end the calculation boils down to commutation relations of the ribbon operators: $F^{\omega,c}_{\xi_1} F^{\sigma,d}_{\xi_2} = \omega(d) \sigma(c)  F^{\sigma,d}_{\xi_2} F^{\omega,c}_{\xi_1}$ for two crossing ribbons $\xi_1$ and $\xi_2$. In the end one finds that
\[
\varepsilon_{(\omega,c),(\sigma,d)} \varepsilon_{(\sigma,d),(\omega,c)} = \overline{\omega}(d) \overline{\sigma}(c),
\]
showing that the charges are indeed abelian anyons.

\subsection{The category of localised endomorphisms}\label{sec:category}
The set of localised endomorphisms has a very rich structure, of which we have given some examples above. Mathematically, it has the structure of a \emph{braided unitary fusion category}.\footnote{Strictly speaking, ``fusion'' implies that there are only finitely many equivalence classes of irreducible objects. This is the case for the quantum double models, but need not be true in general. Dropping that condition does not make a difference in many cases.} The theory of such categories is quite rich, and an active field of study, in particular considering the classification of such categories. There are by now a few texts that provide an accessible entry to the literature. We mention for example~\cite{MR1321145,halvapp,muegerfusion,Wang}, each of which has a different focus. 

We will denote this category by $\Delta$. Its objects are the cone localised and transportable endomorphisms of $\alg{A}$. The morphisms are intertwiners: $T \in \Hom(\rho, \sigma)$ if $T \rho(A) = \sigma(A) T$ for all $A \in \alg{A}$. Note that if $T \in \Hom(\rho,\sigma)$, then $T^* \in \Hom(\sigma,\rho)$, where $T^*$ is the adjoint of $T$. The isomorphisms in this category are unitary operators, so that the we indeed obtain equivalence classes of localised and transportable endomorphisms. An object is irreducible, by definition, if $\Hom(\rho,\rho) \cong \mathbb{C}$. Note that this is just another way of stating that the commutant of $\rho(\alg{A})$ is trivial, hence this coincides with the usual notion of an irreducible representation.

In the previous section we defined a tensor product. This is actually a tensor product in the category, where the trivial endomorphism is the tensor unit. It is in fact a \emph{strict} tensor category: associativity of the tensor product holds on the nose, not only up to isomorphism as is often the case in category theory. For each pair of objects $\rho_1$ and $\rho_2$, the definition of the braid operator $\varepsilon_{\rho_1,\rho_2}$ gives an isomorphism $\varepsilon_{\rho_1,\rho_2} \in \Hom(\rho_1, \rho_2)$. This assignment is in fact natural in both variables, and satisfies the so-called ``braid equations'', which give consistency conditions. This makes $\Delta$ into a \emph{braided} tensor category. It is \emph{symmetric} if $\varepsilon_{\rho_1, \rho_2}^{-1}  = \varepsilon_{\rho_2, \rho_1}$ for all objects $\rho_1, \rho_2$.

Direct sums can be described in the category theory language as well. If $\rho_1$ and $\rho_2$ are objects in $\Delta$, an object $\rho_1 \oplus \rho_2$ is a direct sum if there are $V_i \in \Hom(\rho_i, \rho_1 \oplus \rho_2)$ with $\sum_i V_i V_i^* = I$ and $V_i^* V_j = \delta_{i,j} I$. In algebraic quantum field theory the existence of such direct sums follows from Property B discussed above. In the case of the quantum double model we can explicitly construct the direct sums. Let $\Lambda$ be a cone and suppose that $\rho_1$ and $\rho_2$ are localised in $\Lambda$. Adapting the arguments of~\cite{KMSW,Matsui}, it follows that $\mc{R}_\Lambda := \pi_0(\alg{A}(\Lambda))''$ is an infinite factor, and hence there are isometries $V_i \in \mathcal{R}_\Lambda$ such that $\sum_{i} V_i V_i^* = I$ and $V_i V_j^* = \delta_{i,j I}$. We can then define
\[
	\rho_1 \oplus \rho_2(A) := V_1 \rho_1(A) V_1^* + V_2 \rho_2(A) V_2^*.
\]
Because of locality it follows that $\rho_1 \oplus \rho_2$ is localised in $\Lambda$, and a straightforward calculation shows that this indeed is a direct sum in the categorical sense. Note that the direct sum is only defined up to unitary equivalence. This also explains how the summation in the fusion rules has to be understood, namely as a direct sum in this category. Note that in the abelian models we consider here, there is only one non-zero fusion coefficient, so the direct sum construction is not necessary.

Finally, there is the notion of \emph{duals}. Physically, these can be interpreted as anti-charges. If $\rho$ is an irreducible object of $\Delta$, a dual (or conjugate) is a $\overline{\rho} \in \Delta$ such that the trivial endomorphism $\iota$ of $\alg{A}$ appears exactly once in the direct sum composition of $\rho \otimes \overline{\rho}$ and of $\overline{\rho} \otimes \rho$. Note that this coincides with the notion of a conjugate in terms of the fusion coefficients $N_{ij}^k$ mentioned before. This implies that there is an isometry $R \in \Hom(\iota, \rho \otimes \overline{\rho})$. We assume all objects have conjugates. In this case, it follows that all Hom-sets are \emph{finite-dimensional} vector spaces (over $\mathbb{C}$). It follows from the fusion rules that for the abelian quantum double model, the conjugate of $(\omega,c)$ is $(\overline{\omega}, c^{-1})$.

So far we have constructed different superselection sectors of the abelian quantum double model, and studied their properties. The question remains if there are perhaps additional sectors, that we have so far overlooked. In other words, are there perhaps representations that satisfy Definition~\ref{def:select}, but are not equivalent to one of the representations constructed so far? This question can be answered by adapting techniques from rational conformal nets~\cite{MR1838752}. In particular, one can consider two disjoint cones $\Lambda_1$ and $\Lambda_2$, and look at the von Neumann algebra $\pi_0(\alg{A}(\Lambda_1 \cup \Lambda_2))''$. While Haag duality holds for a single cone, for a \emph{pair} of cones it generally does not hold anymore. This means that $\pi_0(\alg{A}(\Lambda_1 \cup \Lambda_2))'' \subset \pi_0(\alg{A}( (\Lambda_1 \cup \Lambda_2)^c)''$ in general is a proper inclusion (it is an irreducible subfactor, in fact). Essentially, the bigger algebra also contains the intertwiners or charge transporters that move a charge from one cone to the other. Hence by studying how much bigger the algebra is, one can learn something about the number of sectors in the theory. This relative size can be quantified by the Jones-Kosaki-Longo index $[\pi_0(\alg{A}( (\Lambda_1 \cup \Lambda_2)^c)'': \pi_0(\alg{A}(\Lambda_1 \cup \Lambda_2))'']$, which in turn gives an upper bound on the number of sectors~\cite{jklindex}. For the toric code this number is four, and hence it follows that we have indeed construct all sectors: the ground state sector, and the ones corresponding to the pairs $(\iota, g)$, $(\sigma, e)$, $(\sigma,g)$, where $\sigma$ is the sign representation of $\mathbb{Z}_2$ and $g$ is the only non-trivial element in the group. This leads to the following conclusion:

\begin{theorem}
	The category of localised and transportable endomorphisms of the toric code is equivalent as a braided fusion category to category $\operatorname{Rep}_f \mc{D}(G)$ of finite dimensional representations of the quantum double $\mc{D}(\mathbb{Z}_2)$.
\end{theorem}
It should be rather straightforward (but tedious) to extend this result to all abelian quantum double models.

There is one aspect that we have not mentioned so far. The category of the quantum double model is \emph{modular}. Modular tensor categories are fusion categories with an additional non-degeneracy property: it is called modular if it has trivial braided centre. That is, suppose that $\rho$ is irreducible, and $\varepsilon_{\rho,\sigma} \varepsilon_{\sigma,\rho} = I$ for all irreducible $\sigma$. If this implies that $\rho$ is the trivial object, then the category is modular~\cite{Mueger03,Rehren90}. In this sense, modular categories are ``as far away'' as possible from being symmetric, where $\varepsilon_{\rho,\sigma} \varepsilon_{\sigma,\rho}$ is always the identity.

There is an equivalent condition that is related to Verlinde's matrix $S$ in conformal field theory~\cite{Verlinde88}. In the categorical setting it can be defined as follows. One can show that the duality allows one to define a \emph{trace} on morphisms in the category. A matrix $S$ is than defined by having entries $S_{ij} = \operatorname{tr}(\varepsilon_{\rho_i, \rho_j} \varepsilon_{\rho_j, \rho_i})$, where $\rho_i$ is a set of representatives of the irreducible objects. A category is modular if and only if this matrix $S$ is invertible. It also allows us to explain the name modular: together with a matrix $T$ that can also be canonically defined, the matrices $S$ and $T$ induce a projective representation of the modular group $\mathrm{SL}_2(\mathbb{Z})$. In modular categories it turns out that there is in fact a relation between the matrix $S$ and the fusion coefficients, given by the Verlinde rule~\cite{Verlinde88}:
\[
	N_{ij}^k = \sum_{r} \frac{S_{ij} S_{jr} \overline{S}_{kr}}{S_{1r}^2},
\]
where the label $1$ stands for the trivial object. This shows that that the structure of a modular tensor category is quite rigid, and not any given fusion rule can be realized in some modular category.

Modular tensor categories can be realized as \emph{topological quantum field theories} (TQFTs), see for example~\cite{MR1292673}. One can therefore think, in a sense, of the type of topologically ordered systems that we have considered here as systems that in the low energy limit can be described by a TQFT\@. It should be noted than in general these are really effective theories, giving an effective description of oft-times very complex and poorly understood condensed matter systems. They also play a fundamental role in the field of topological quantum computing~\cite{Wang}, on which we comment briefly in the last section.

\section{Extension to non-abelian models}
So far we have only considered abelian models. A natural question is if the methods can be extended to \emph{non}-abelian quantum spin models. We again consider the quantum double model, but now for a non-abelian (but still finite) group $G$. In that case, there still is a unique translational invariant ground state $\omega_0$. Many of the proofs, however, do not directly carry over from the abelian case. The reason for this will become clear below, but the underlying difficulty is that the irreducible representations of $\mc{D}(G)$ are no longer all one-dimensional. In particular, it is not clear how one could construct endomorphisms describing these ``non-abelian'' charges.

A similar problem appeared in one-dimensional spin chains, with compactly localised charges. There the problem is that the algebra $\alg{A}(\Lambda)$ of such a localization region is finite dimensional, namely the tensor product of finitely many matrix algebras. All endomorphisms of this algebras are in fact automorphisms, and one can show that these cannot have non-abelian statistics. This problem can be circumvented by the methods of Szlach{\'a}nyi and Vecserny{\'e}s~\cite{MR1234107}, and Nill and Szlach{\'a}nyi~\cite{MR1463825}. Instead of looking at endomorphisms, they look at \emph{amplimorphisms}, that is, morphisms $\chi: \alg{A} \to M_n(\alg{A})$ for some integer $n$. One can then do a study of these endomorphisms in the spirit of the DHR theory.

In the present case we are interested in two dimensional systems with localization in \emph{cones}. The cone algebras are certainly not finite dimensional, so the obstruction of the only endomorphisms being automorphisms does not play a role here.\footnote{To be a bit more precise, one actually needs that $\pi_0(\alg{A}(\Lambda))''$ is not a factor of Type I, that is, not of the form $\alg{B}(\mc{H})$ for some Hilbert space $\mc{H}$. For the quantum double models that is the case.} Nevertheless, as mentioned it is not easy to explicitly construct examples of the endomorphisms. This is where the amplimorphisms come in. The strategy is to mimic the amplimorphism construction in~\cite{MR1234107} in the context of cone-localised charges, to construct representatives of the different charge classes, and will then show how we can go back to the usual setting of cone-localised endomorphisms. The amplimorphism description is much more explicit, making it possible to explicitly calculate intertwiners, fusion rules, \emph{et cetera}. Here we will mainly restrict to the construction of the different sectors. Combining the techniques developed in the abelian case with the amplimorphism results in~\cite{MR1234107} should enable one to completely solve the model.

The ribbon operators again play a fundamental role. We will use the same notation as in Section~\ref{sec:qdouble}, e.g. $r$ will always be a fixed representative in a conjugacy class. Let $\xi$ be a fixed ribbon, and recall that one can choose a basis of the ribbon operators acting on $\xi$ in terms of the irreducible representations of $\mc{D}(G)$. For a fixed representation there is a multiplet of ribbon operators $F^{IJ}_\xi$, where $I = (i_1, i_2)$ and $J = (j_1, j_2)$. The first indices run over the elements of the corresponding conjugacy class $C$, while the second runs over the dimension of the group. These multiplets satisfy certain completeness relations, namely
\begin{equation}
	\sum_{I} (F^{IJ}_\xi)^* F_{\xi}^{IK} = \delta_{JK} I, \quad \sum_{J} F^{IJ}_\xi (F^{KJ}_\xi)^* = \delta_{IK} I,
	\label{eq:multiplet}
\end{equation}
where the summations are over all pairs $(i_1, i_2)$ and the $I$ on the right hand side is the unit of $\alg{A}$. These equations can be verified by a simple calculation using the definitions and the algebraic relations for $F_{\xi}^{h,g}$. In~\cite{MR1234107} these $F^{IJ}$ are called \emph{irreducible} and \emph{complete} multiplets, with the difference that here we have not defined an action $\gamma_a$ of $\mc{D}(G)$ acting on these multiplets. Nevertheless, they do transform under an irreducible representation of $\mc{D}(G)$, cf.\ equation~(B69) of~\cite{PhysRevB.78.115421}.

Recall that for the abelian model, a ribbon operator acting on a large ribbon is just the product of two ribbon operators acting on smaller ribbons, but related to the same irreducible representation. If the irreducible representation is not one-dimensional any more, this is no longer true. Nevertheless, a suitable analogue of equation~\eqref{eq:ribbonop} still holds:
\begin{lemma}
\label{lem:multidecompose}
	Choose a pair $(C,\rho)$ of a conjugacy class and an irreducible representation of the centralizer $Z_G(r)$, where $r$ is as explained in Section~\ref{sec:qdouble}. Let $\xi = \xi_1 \xi_2$ be a ribbon that is decomposed into two ribbons. The corresponding multiplets will be denoted by $F_{\xi}^{IJ}$ and $F_{\xi_i}^{IJ}$ respectively. Then we have the following relation
\begin{equation}
	F_{\xi}^{IJ} = \sum_{K} F^{IK}_{\xi_1} F^{KJ}_{\xi_2},
	\label{eq:multidecom}
\end{equation}
where the sum is over all pairs $K = (k_1, k_2)$
\end{lemma}
\begin{proof}
First write out the right hand side of equation~\eqref{eq:multidecom} in terms of the elementary ribbon operators $F_{\xi_i}^{h,g}$, where we set $I = (i_1, i_2)$, and similarly for $J$ and $K$:
\[
\sum_{g,h \in Z_G(r)} \sum_{k_1 = 1}^{|C|} \sum_{k_2=1}^{\dim(\rho)} \overline{\rho}_{i_2 k_2}(g) \overline{\rho}_{k_2 j_2}(h) F_{\xi_1}^{\inv{c}_{i_1}, q_{i_1} g \inv{q}_{k_1}}  F_{\xi_2}^{\inv{c}_{k_1}, q_{k_1} h \inv{q}_{j_1}}.
\]
Since $\rho$ is a representation, the summation over $k_2$ yields a term $\overline{\rho}(gh)$. After a substitution $h \mapsto \inv{g} h$ we obtain
\[
\sum_{g,h \in Z_G(r)} \sum_{k_1 = 1}^{|C|} \overline{\rho}_{i_2 j_2}(h) F_{\xi_1}^{\inv{c}_{i_1}, q_{i_1} g \inv{q}_{k_1}}  F_{\xi_2}^{\inv{c}_{k_1}, q_{k_1} \inv{g} h \inv{q}_{j_1}}.
\]
As remarked before, every element $s \in G$ can be written uniquely in the form $s = n q_i$ for some $i$ and $n \in Z_G(r)$. Hence the summation over $k_1$ and $g$ can be replaced by a summation over $s \in G$. More precisely, we set $s = g \overline{q}_{k_1}$. Note that $c_{k_1} = q_{k_1} r \inv{q}_{k_1} = s r \inv{s}$. With this observation the expression above reduces to
\[
\sum_{h \in Z_G(r)} \sum_{s \in G} \overline{\rho}_{i_2 j_2}(h) F_{\xi_1}^{\inv{c}_{i_1},q_{i_1} s} F_{\xi_2}^{\inv{s}\,\inv{r} s, \inv{s} h \inv{q}_{j_1}} = F^{IJ}_{\xi},
\]
where the equality follows after a substitution $s \mapsto \inv{q}_i s$ and with the help of equation~\eqref{eq:ribbonop}. This completes the proof.
\end{proof}

Analogously to the automorphisms for the abelian models, we now define linear maps $\chi_{IJ}(A)$ of $\alg{A}$. Choose again a semi-infinite ribbon $\xi$ and let $F^{IJ}_{\xi_n}$ denote the corresponding multiplets, where $\xi_n$ is the first part of the ribbon, consisting of $n$ triangles. Then we set for any local observable $A$:
\begin{equation}
	\chi_{IJ}(A) := \lim_{n \to \infty} \sum_{K} F^{IK}_{\xi_n} A \left(F_{\xi_n}^{JK}\right)^*.
	\label{eq:amplicomponents}
\end{equation}
Note that (assuming for the moment that the expression converges) this defines a linear map defined on a dense subset of $\alg{A}$. Since it is bounded, it can be extended to $\alg{A}$. This extension will also be denoted by $\chi_{IJ}$. So it remains to be shown that the expression indeed converges (in norm, even). The main idea is similar as in the abelian case, only here we have to use Lemma~\ref{lem:multidecompose}.

\begin{lemma}
	\label{lem:amplimat}
	Let $\chi_{IJ}$ be as above. Then the limit on the right hand side of equation~\eqref{eq:amplicomponents} converges. We have the following properties:
\begin{enumerate}
	\item \label{it:local}for $A \in \alg{A}_{loc}$, $\chi_{IJ}(A) = \sum_{K} F_{\xi_N}^{IK} A \left(F_{\xi_N}^{JK}\right)^*$ for $N$ big enough;
	\item $\chi_{IJ}(I) = \delta_{IJ} I$;
	\item $\chi_{IJ}(A) = \delta_{IJ} A$ if $A$ is localised away from the ribbon;
	\item $\chi_{IJ}(AB) = \sum_{K} \chi_{IK}(A) \chi_{KJ}(B)$;
	\item $\chi_{IJ}(A)^* = \chi_{JI}(A^*)$.
\end{enumerate}
\end{lemma}
\begin{proof}
	We show the first property. The others then follow straightforwardly using orthogonality and completeness for the multiplets, as well as the definition of $\chi_{IJ}$. Consider $A \in \alg{A}_{loc}$. Let $N$ be such that $\supp(A) \cap (\xi_n \setminus \xi_N) = \emptyset$ for all $n \geq N$. The idea is to decompose the ribbon $\xi_n$ as $\xi_n = \xi_N \widehat{\xi}$, where $\widehat{\xi} = \xi_n \setminus \xi_N$. Let us now write $\chi_{IJ}^{n}(A)$ for $\sum_{K} F^{IK}_{\xi_n} A \left(F_{\xi_n}^{JK}\right)^*$, and set $\xi_1 = \xi_N$, $\xi_2 = (\xi_n \setminus \xi_N)$. By Lemma~\ref{lem:multidecompose}, locality, and equations~\eqref{eq:multiplet} we get
\[
\begin{split}
	\chi_{IJ}^n(A) &= \sum_{K} \sum_{L} \sum_{M} F_{\xi_1}^{I L} F_{\xi_2}^{L K} A (F_{\xi_1}^{J M} F_{\xi_2}^{M K})^* \\
				   &= \sum_{K} \sum_{L} \sum_{M} F_{\xi_1}^{I L} A (F_{\xi_1}^{J M})^* F_{\xi_2}^{L K} (F_{\xi_2}^{M K})^* \\
				   &= \sum_{L} \sum_{M} F_{\xi_1}^{I L} A (F_{\xi_1}^{J M})^* \delta_{LM} \\
				   &= \chi_{IJ}^N(A).
\end{split}
\]
From this it is clear that the limit in equation~\eqref{eq:amplicomponents} converges for operators $A \in \alg{A}_{loc}$. As mentioned the other properties are easy to verify.
\end{proof}

Note that the properties stated in the Lemma are precisely those that one needs to define an amplimorphism. Note that there are $n = |C| \dim(\rho)$ pairs $I = (i,j)$. We then define a map $\chi : \alg{A} \to M_n(\alg{A})$ by setting $[\chi(A)]_{IJ} = \chi_{IJ}(A)$. It follows that $\chi$ is an amplimorphism. In addition, it is localised in any cone $\Lambda$ that contains the ribbon $\xi$ used in the definition of $\chi_{IJ}$, in the sense that $\chi(A)$ is the matrix with entries $A$ on the diagonal and otherwise zeros, if $A \in \alg{A}(\Lambda^c)$. It remains to show that the amplimorphisms are transportable. Here we show that the charges can be transported over a \emph{finite} region.

\begin{lemma}
The amplimorphisms $\chi$ constructed above are transportable over finite distances.
\end{lemma}
\begin{proof}
	Fix a semi-infinite ribbon $\xi$. We demonstrate how we can move the charge at the endpoint of the ribbon around. That is, let $\widehat{\xi}$ be a ribbon, such that $\widehat{\xi}\xi$ is again a semi-infinite ribbon. Define $V \in M_n(\alg{B}(\mc{H}))$, where $n = |C| \dim(\rho)$, by having entries $V_{IJ} = F_{\widehat{\xi}}^{IJ}$. From equations~\eqref{eq:multiplet} it follows that $V$ is unitary. The claim is that $V \chi(A) V^* = \widehat{\chi}(A)$, where the amplimorphism in the right hand side is defined with respect to the ribbon $\widehat{\xi}$. This can be verified for local operators $A$, by carrying out the matrix multiplication, and using Lemma~\ref{lem:multidecompose} together with~\ref{it:local}. of Lemma~\ref{lem:amplimat}. 
\end{proof}

The case of cone transportability is more complicated. Because of the lemma it is enough to consider two semi-infinite ribbons $\xi_1$ and $\xi_2$ starting at the same site. Let $\xi_{i,n}$ denote the ribbon consisting of the first $n$ triangles. For each $n$, choose a ribbon $\widehat{\xi}_n$ that goes from the endpoint of $\xi_{2,n}$ to the endpoint of $\xi_{1,n}$ in such a way that as $n$ goes to infinity, so does the distance of the ribbon to the endpoint of $\xi_1$. Now define a unitary operator $V_n = V_{\xi_{2,n} \widehat{\xi}_n} V_{\xi_{1,n}}^*$, where $V_{\xi_{i,n}}$ is the unitary obtained from the multiplet $F_{\xi_{i,n}}^{IJ}$. Now if $A$ is local, it follows that $V_n \chi_1(A) = \chi_2(A) V_n$ for all $n$ large enough. This can be seen by the argument in the proof of the above lemma.

This gives a uniformly bounded sequence of operators, since each of them is unitary. By the compactness of the unit ball in the weak operator topology, there is a subnet that converges to some operator $V$. Since multiplication on the right with a fixed $\chi(A)$ is weakly continuous, it follows that $V$ intertwines $\chi(A)$ and $\widehat{\chi}(A)$. The problem remains to show that $V$ is unitary. There are different ways that one might achieve this. For example, one could first try to show that $\chi$ and $\widehat{\chi}$ are irreducible, so that $\Hom(\chi,\widehat{\chi})$ must be either zero or one-dimensional. If $V$ is non-zero, it then follows that one can choose $V$ to be unitary. The other option is to realize both representations as the GNS representation of the same state. The vector state with an $\Omega$ in the first component (and otherwise zero) is a good candidate. By the independence of the ribbon operators on the exact choice of ribbon, this leads to the same state in both representations. If one can show that this vector is in fact cyclic, the proof is complete. We will leave this issue open for now.

It is natural to look at the amplimorphisms in the ground state representation, that is, look at $(\pi_0 \otimes \id) \circ \chi$. This amounts to applying the ground state representation to each matrix element. In fact, that is what we have been doing implicitly above. Note that by localization of the amplimorphisms, for observables outside the localization region, this representation looks like $n$ copies of the ground state representation. Hence it is natural to adapt the selection criterion~\eqref{eq:select} a bit to allow for this case. It turns out that in the end this does not really matter, and one can go back from the amplimorphism picture to cone-localised endomorphisms (or representations). This is the content of the next theorem.
\begin{theorem}
	Suppose that $\pi_0$ satisfies Haag duality for cones. Let $\pi$ be a representation of $\alg{A}$ and $n$ be a positive integer such that for some cone $\Lambda$, we have
\begin{equation}
	n \cdot \pi_0 \upharpoonright \alg{A}(\Lambda^c) \cong \pi \upharpoonright \alg{A}(\Lambda^c),
	\label{eq:nselect}
\end{equation}
where $n \cdot \pi_0$ is the direct sum of $n$ copies of the representation. Then the following hold:
\begin{enumerate}
	\item\label{it:ampli}There is an amplimorphism $\chi: \alg{A} \to M_n(\alg{A}^{\Lambda_a})$, localised in $\Lambda$, such that we have $(\pi_0 \otimes \id) \circ \chi \cong \pi$.
	\item\label{it:endo}There is a morphism $\rho: \alg{A} \to \alg{A}^{\Lambda_a}$ such that $\pi_0 \circ \rho \cong \pi$ when restricted to $\alg{A}$.
\end{enumerate}
Note that the unitary equivalence in the second point can be used to map intertwiners between amplimorphisms to intertwiners between morphisms. In particular, if $\chi$ is transportable, so is the corresponding morphism, and $\rho$ can be extended to an endomorphism of $\alg{A}^{\Lambda_a}$.
\end{theorem}
\begin{proof} (\ref{it:ampli}) Let $U: \mathcal{H}_\pi \to \bigoplus_{i=1}^n \mathcal{H}_0$ be the unitary setting up the equivalence and suppose that $A \in \alg{A}(\Lambda)$. For $A \in \alg{A}$, define $\chi(A) := U \pi(A) U^*$. Note that $\chi(A) \in \alg{B}(\bigoplus_{i=1}^n \mc{H}_0) \cong M_n(\alg{B}(\mc{H}_0))$, so that the matrix elements $\chi_{ij}(A) \in \alg{B}(\mc{H}_0)$. Now let $B \in \alg{A}(\Lambda^c)$ and $A \in \alg{A}(\Lambda)$. Note that from equation~\eqref{eq:nselect} it follows that $\chi(B) = \diag(\pi_0(B), \dots, \pi_0(B))$.
	
	Remark that $A$ and $B$ commute by locality,  so that $\chi(A)\chi(B) = \chi(AB) = \chi(B)\chi(A)$. Writing out the definitions and comparing matrices element wise, it follows that $\pi_0(B) \chi_{ij}(A) = \chi_{ij}(A) \pi_0(B)$. Hence $\chi_{ij}(A) \in \pi_0(\alg{A}(\Lambda^c))' = \pi_0(\alg{A}(\Lambda))''$, by Haag duality. Since the algebra on the right hand side is contained in the auxiliary algebra $\alg{A}^{\Lambda_a}$, and $\chi$ acts trivially on $\alg{A}(\Lambda^c)$, it follows that $\chi: \alg{A} \to M_n\left(\alg{A}^{\Lambda_a}\right)$ is an amplimorphism.

	(\ref{it:endo}) Since $\pi_0(\alg{A}(\Lambda))''$ is an infinite factor, one can find isometries $V_i$, $i=1, \dots n$ generating a Cuntz algebra~\cite{MR0467330}. That is, they satisfy $V_i^* V_j = \delta_{ij} I$ and $\sum_{i = 1}^n V_i V_i^* = I$. Write $\chi$ for the amplimorphism obtained in part~\eqref{it:ampli}, and $\chi_{ij}(A)$ for its matrix elements when evaluated in $\alg{A}$. Define a map $\rho : \alg{A} \to \alg{A}^{\Lambda_a}$ by
\[
	\rho(A) = \sum_{i,j = 1}^n V_i \chi_{ij}(A) V_j^*.
\]
Suppose that $A,B \in \alg{A}$. A straightforward calculation then shows
\begin{align*}
	\rho(A)\rho(B) &= \sum_{i,j,k,l} V_i \chi_{ij}(A) V_j^* V_k \chi_{kl}(B) V_l^* \\
	&= \sum_{i,j,l} V_i \chi_{ij}(A) \chi_{jl}(B) V_l^* \\
	&= \rho(AB).
\end{align*}
Moreover, $\rho(A^*) = \rho(A)^*$ since $\chi_{ij}(A^*) = \chi_{ji}(A)^*$. If $B \in \alg{A}(\Lambda^c)$, then $\chi_{ij}(B) = \delta_{ij} B$ and $B \in \alg{A}(\Lambda)'$, hence $\rho(B) = B$ and $\rho$ is localised in $\Lambda$.

Next we show that $\pi_0 \circ \rho$ is unitarily equivalent to $\pi_0 \otimes \chi$. To this end, identify $\mc{H}_0 \otimes \mathbb{C}^n$ with $\mc{H} = \bigoplus_{i=1}^n \mc{H}_0$. Define a map $U: \mc{H} \to \mc{H}_0$ by setting
\[
U (\psi_1 \oplus \dots \oplus \psi_n) = \sum_{i=1}^n V_i \psi_i.
\]
Using the properties of the $V_i$ it is easy to check that $U$ preserves the inner product (and hence is an isometry and well-defined). Since $\sum_{i=1}^n V_i V_i^* = I$ it also has dense range, hence $U$ is unitary.
\end{proof}

A few remarks are in order at this point. First of all, the condition of Haag duality is only used to obtain an amplimorphism from a representation satisfying equation~\eqref{eq:nselect}. Without Haag duality, one can still obtain an morphism $\rho$ such that $\pi_0 \circ \rho$ and $\pi_0 \otimes \chi$ are unitarily equivalent. However, we then have little control over the range of the morphism, and we cannot extend it to an endomorphism of the auxiliary algebra without any additional information. 

The theorem shows that in principle one can restrict to the study of localised and transportable endomorphisms. Nevertheless, it can be very helpful to look at amplimorphisms as well. One reason is that it may be easier to construct such amplimorphisms explicitly in concrete models, as we have done above. They also provide more information on the symmetries of the model. In particular, the vector space $\mathbb{C}^n$ in $\mc{H}_0 \otimes \mathbb{C}^n$ carries a representation of the symmetry algebra, through the symmetry transformations of the multiplets $F_\xi^{IJ}$. This plays an important role in the analysis of the superselection sectors in~\cite{MR1463825,MR1234107}. We expect that the methods used there to study the category of amplimorphisms (and hence, by the theorem above, the category of localised endomorphisms). In particular, to define fusion and braiding. This leads to the following conjecture:

\begin{conjecture}
Let $G$ be a finite group. Then the category $\Delta$ of localised and transportable endomorphisms of the quantum double model is equivalent to the representation category $\operatorname{Rep}_f \mathcal{D}(G)$ of the quantum double $\mc{D}(G)$.
\end{conjecture}

In~\cite{phdthesis} we obtained through computer algebra the fusion rules of $\mc{D}(S_3)$ by taking the composition of certain positive maps. These results are consistent with the approach here, since the positive maps there are the traces of the amplimorphisms we have defined here. The composition of these positive maps yields the trace of the amplimorphism $\chi_1 \otimes \chi_2$, where the tensor product is defined as in~\cite{MR1234107}.

To study this conjecture it would be helpful to go from endomorphisms to amplimorphisms. An important question is also if there is a ``canonical'' way to obtain an amplimorphism from an endomorphism. Of course, there is always the ``trivial'' way, since and endomorphism $\rho$ gives rise to an amplimorphism $\widehat{\rho} : \alg{A} \to \alg{A} \otimes \mathbb{C}$. In general, let $T_1, \dots T_n \in \alg{B}(\mc{H})$ be such that $T_i^* T_j = \delta_{i,j} I$ and $ \sum_{k=1}^n T_k T_k^* = I$. Then it is easy to check that $[\chi(A)]_{ij} := T_i^* \rho(A) T_j$ defines an amplimorphism. Hence to have a meaningful equivalence between the category of amplimorphisms and of localised morphisms, we would have to impose some additional conditions on the amplimorphisms. One of them is that they transform in the right way, as explained above. To find this symmetry one could look for an analogue of the Doplicher-Roberts theorem~\cite{DR90}, which gives a group symmetry for bosonic/fermion sectors. In general, one could only expect a so-called weak Hopf-algebra symmetry~\cite{Ostrik}.

We have outlined here how one could proceed with an analysis of the superselection structure of the non-abelian quantum double model. Although the analysis is not complete, we hope that it is a helpful starting point. It should be remarked that here we have provided an explicit model whose ground space representation should lead to non-abelian charges. This should be contrasted with~\cite{MR1463825,MR1234107}, where the existence of such a representation is taken as an assumption.

\section{An application: topological quantum computing}
There are different reasons why there has been a huge interest in topologically ordered systems in recent years. One of the reasons is that they provide examples of new phases of matter, that go beyond the Landau paradigm of symmetry breaking. This is not just of theoretical interest -- these phases really exist in nature. There is also experimental evidence for the existence of quasi-particle excitations with anyonic excitations, see for example~\cite{MajoranaFermions}. A good theoretical understanding is therefore very welcome.

Here we focus on another aspect that has sparked the interest of the quantum computation community. One of the goals of quantum computation is to use the full power of quantum mechanics to solve computationally hard problems, for which a computation on a usual, classical computer is infeasible. To illustrate this one can think of a simple spin-1/2 quantum system, with Hilbert space $\mathbb{C}^2$. The dimension of $n$ copies of such a system scales as $2^n$, so if one wants to simulate the whole Hilbert space of a $n$-particle system, one quickly runs out of memory in a classical computer. The idea behind quantum computing, which goes back to Feynman~\cite{FeynmanQC}, is to use the laws of quantum mechanics to solve computationally complex problems, or even simulate other quantum systems.

This is not the place for a full-fledged introduction to quantum computing (for this Nielsen and Chuang's~\cite{MR1796805} is a good start), but let us summarise the main points. The ``quantum memory'' of a quantum computer is modelled by a (generally finite dimensional) Hilbert space, ususally of the form $\mathcal{H} = (\mathbb{C}^2)^{\otimes n}$ for some $n$. A computation then consists of three steps:
\begin{enumerate}
	\item Initialise the system in a known state;
	\item Perform a unitary operation on the system to implement the algorithm;
	\item Measure the result (and if necessary, repeat to get statistics).
\end{enumerate}
This encompasses classical computing. To see this, we can reformulate each computational problem in the calculation of a function $f: \{0,1\}^n \to \{0,1\}^n$, where for simplicity we take $f$ to be injective. Now let $\ket{0}, \ket{1}$ be a basis of $\mathbb{C}^2$. Then we can identify each $(x_i) \in \{0,1\}^n$ with a basis vector $\ket{x} := \ket{x_1} \otimes \cdots \otimes \ket{x_n}$. We can then define a map $U_f$ by $U_f \ket{x} = U_f \ket{f(x)}$, which is unitary because $f$ is injective. Hence we can initialize the system in a known state $\ket{x}$, apply $U_f$, and measure to learn something about $\ket{f(x)}$. This is where quantum mechanics comes in: it allows us to act with $U_f$ on a \emph{superposition}, say of the form $\frac{1}{\sqrt{2^n}} \sum_x \ket{x}$, so that we can learn something about \emph{all} values $f(x)$ in a \emph{single} operation. This is not possible on a classical computer.

Although the main idea is quite simple, it turns out to be very difficult to implement this in physical systems in a reliable and scalable way. This is were topological quantum computing comes in~\cite{MR1951039,MR2443722}. We have already briefly discussed the problem of storing quantum information over an extended period of time. Here we focus on the computation part, that is, processing this quantum information to implement an algorithm. As mentioned above, this amounts to acting with a certain unitary operator on the system. In practice this could work, for example, by coupling the system to some external magnetic field for a period of time, and let it evolve. The problem is that it generally is very difficult to exactly perform the unitary that you want, and not something slightly different (because the magnetic field is left on too long, for example). The idea is therefore to use ``topological'' operations to implement the necessary unitary. A small disturbance should not change the topological property, and hence have no effect on the computation. The braiding operation of two anyons for example is independent of the path that the anyons take (as long as the paths don't cross).

To use this idea, the quantum information on which we want to operate has to be encoded using anyons. Let us consider a non-abelian anyon $\rho$, which for the sake of simplicity we assume to be self-dual: $\rho = \overline{\rho}$. We can then take $n$ copies of it, that is, consider $\rho^{\otimes n}$. The idea is to create these by pulling pairs of them from the vacuum, which is possible because $\rho$ is its own anti-charge. Because the anyon is non-abelian, this can actually be done in different ways. More precisely, the possible states are described by $\Hom(\iota, \rho^{\otimes n})$, where $\iota$ is again the trivial charge. This vector space can be given the structure of a Hilbert space. It is this space that we use to encode the qubits in. It should be noted that in general $\Hom(\iota, \rho^{\otimes})$ does \emph{not} have a nice decomposition as the tensor product of $n$ copies of some Hilbert space. Nevertheless, one can embed the state space of a number of qubits into this space. The dimension of $\Hom(\iota, \rho^{\otimes n})$ grows exponentially in $n$.

Step 1 in the quantum computation scheme can now be accomplished by pulling charges from the vacuum in a suitable way. To implement unitary operations on the qubits, we note that there is a natural representation of the braid group $B_n$ on $\Hom(\iota, \rho^{\otimes n})$: if $T \in \Hom(\iota, \rho^{\otimes n})$, we can consider
\[
	\pi(b_i) T := (I \otimes \cdots \otimes \varepsilon_{\rho,\rho} \otimes \cdots \otimes I) \circ T,
\]
where $b_i$ is the generator of the braid group that swaps the $i$-th and $(i+1)$-th strand, while on the right hand side, the $\varepsilon_{\rho,\rho}$ term acts on the $i$-th and $(i+1)$-th tensor factors. Note that the result is again in $\Hom(\iota, \rho^{\otimes n})$, hence this gives a unitary rotation on the encoded qubits, and hence allows us to implement (a part of) a quantum computation. 

This is not the end of the story, because it is not clear if every unitary on the encoded qubits can be obtained in this way. It is enough for $\pi(B_n)$ to generate a dense subset of the unitaries. If this is the case for a certain anyon model, it is said to be \emph{universal}, because it means that each quantum algorithm can in principle be implemented on it. Kitaev's quantum double model is universal for a wide range of non-abelian groups~\cite{Mochon1,Mochon2}, but in general this is a rather special feature of a model. If a model is not universal, one can supplement the braiding operations with other, non-topological operations to be able to implement each unitary operation. Even if the model is universal, one still has to find out which combination of braidings one has to do to obtain a certain unitary. Luckily, this can be found efficiently using the Solovay-Kitaev theorem~\cite{MR1796805}. This is worked out explicitly for a simple anyon model, the so called Fibonacci model, in~\cite{Bonesteel}.

Finally, after the appropriate braidings have performed, it is time to measure the outcome. That is, we have to find out which state in $\Hom(\iota, \rho^{\otimes n})$ the system is in. This is done by fusing some of the anyons again. Recall that the fusion rules are of the form $\rho \otimes \rho = \sum_k N^k_{\rho\rho}\rho_k$. The integers $N^k_{\rho\rho}$ essentially says in how many ways the fusion of two $\rho$-anyons can lead to an anyon $\rho_k$. It are precisely these different ways to fuse $n$ anyons to the vacuum that label the states. Hence, by fusion some anyons and observing the outcome (which amounts to doing a charge measurement), we can obtain statistics on which state the system was in. The charge measurement is where it becomes important to have modularity of the category: this allows us to distinguish the charges by pulling pairs of charges from the vacuum, move one of the anyons around the region for which we want to determine the charge, and fuse to the vacuum again and observe if there is anything left or not.

It should be noted that while topological quantum computing has advantages with respect to the stability of the operations, there are also drawbacks. For example, in practice it may not be so easy to physically move the anyons around, especially over larger distances. This may be circumvented by measurement based quantum computation. There, the braiding operations are mimicked by doing a series of measurements~\cite{MBQC}. One could also restrict to more local operations, although this will mean that the anyon models are no longer universal~\cite{localtqc}. 

To conclude, we have seen that there is a rich class of so-called topologically ordered states, for which methods of local quantum physics provide useful tools to study such systems. Besides the possible applications to quantum computing that we have mentioned, there are also very interesting condensed matter and mathematical aspect related to such phases. Finally, also on the experimental side the field is very active. It would be good to see if the tools of local quantum physics can be further employed to advance progress in this multi-disciplinary field.
\\

\noindent\textbf{Acknowledgements:}
The author wishes to thank Courtney Brell for helpful comments and discussions and Leander Fiedler for collaboration on~\cite{haagdouble}. This work is supported by the Dutch Organisation for Scientific Research (NWO) through a Rubicon grant and partly through the EU project QFTCMPS and the cluster of excellence EXC 201 Quantum Engineering and Space-Time Research	

% references
\input{referenc}

\end{document}

%% file: referenc.tex
%%%%%%%%%%%%%%%%%%%%%%%% referenc.tex %%%%%%%%%%%%%%%%%%%%%%%%%%%%%%
% sample references
% %
% Use this file as a template for your own input.
%
%%%%%%%%%%%%%%%%%%%%%%%% Springer-Verlag %%%%%%%%%%%%%%%%%%%%%%%%%%
%
% BibTeX users please use
% \bibliographystyle{}
% \bibliography{}
%